\documentclass[journal, onecolumn, 12pt]{IEEEtran}
\usepackage{setspace}
\usepackage{tikz}
\usepackage{amsfonts}
\usepackage{indentfirst}
\usepackage{amsmath,amsthm}
\usepackage{amssymb}
\usepackage{color}
\usepackage{graphicx}
\usepackage{multirow,makecell,diagbox}
\usepackage[para]{threeparttable}
\usepackage{caption}
\usepackage{subfigure}
\usepackage{cite}
\usepackage{epstopdf}
\newtheorem{Theorem}{Theorem}
\newtheorem{Corollary}{Corollary}
\newtheorem{Definition}{Definition}

\newtheorem{Remark}{Remark}
\newtheorem{Lemma}{Lemma}
\newtheorem{Example}{Example}

\evensidemargin=\oddsidemargin
\begin{document}
\doublespacing{}
\title{Low-PMEPR Preamble Sequence Design for Dynamic Spectrum Allocation in OFDMA Systems}

\author{Yajing Zhou,
Zhengchun Zhou,
Zilong Liu,~\IEEEmembership{Member,~IEEE,} Pingzhi Fan,~\IEEEmembership{Fellow,~IEEE,} and
Yong Liang Guan,~\IEEEmembership{Senior Member,~IEEE,}

\thanks{Y. Zhou is with the School of Information Science and Technology, Southwest Jiaotong University,
Chengdu, 611756, China. E-mail: zhouyajing@my.swjtu.edu.cn. }
\thanks{Z. Zhou is with the School of Mathematics, Southwest Jiaotong University,
Chengdu, 611756, China. E-mail: zzc@swjtu.edu.cn. He is also  with the State Key Laboratory of Integrated
Services Networks, Xidian University, Xi¡¯an, 710071, China.}
\thanks{Z. Liu is with School of Computer Science and Electronic Engineering, University of Essex, UK. E-mail: zilong.liu@essex.ac.uk}
\thanks{P. Fan is with the Institute of Mobile Communications, Southwest Jiaotong
University, Chengdu, 611756, China. E-mail: pzfan@swjtu.edu.cn.}

\thanks{Y. L. Guan is with the School of Electrical and Electronic Engineering,
Nanyang Technological University, Singapore, 639798. E-mail:eylguan@ntu.edu.sg.}}
\maketitle

\date{}
\begin{abstract}
Orthogonal Frequency Division Multiple Access (OFDMA) with Dynamic spectrum allocation (DSA) is able to provide a wide range of data rate requirements. This paper is focused on the design of preamble sequences in OFDMA systems with low peak-to-mean envelope power ratio (PMEPR) property in the context of DSA. We propose a systematic preamble sequence design which gives rise to low PMEPR for possibly non-contiguous spectrum allocations. With the aid of Golay-Davis-Jedwab (GDJ) sequences, two classes of preamble sequences are presented. We prove that their PMEPRs are upper bounded by 4 for any DSA over a chunk of four contiguous resource blocks.

\end{abstract}

\section{Introduction}

Orthogonal Frequency Division Multiple Access (OFDMA) is a high-rate multiple access scheme where different users are allocated with non-overlapping spectral bands. Due to its resilience to intersymbol interference and low-complexity equalization at the receiver end, OFDMA has attracted significant attention over the past decades. For example, OFDMA has been adopted in IEEE 802.16 standard \cite{IEEE-4} and Long Term Evolution (LTE) downlink \cite{Sesia-LTE-11}. The multicarrier transmission nature of OFDMA implies that it may suffer from high peak-to-mean envelope power ratio (PMEPR) which could result in distorted transmitted signals and reduced communication range \cite{Goldsmith-05}. Different from traditional OFDM system operated over a dedicated contiguous spectral band, dynamic spectrum allocation (DSA) is employed in OFDMA to accommodate instantaneous network conditions and different requirements of quality of service. In DSA, a large contiguous spectral band is divided into several resource blocks (RBs), and any user, after a request-and-grant random access procedure, may be given one or more RBs which are contiguous or non-contiguous. Note that a RB, which is comprised of several contiguous subcarriers, is the smallest spectrum allocation unit in OFDMA systems.

The objective of this paper is to design preamble sequences with low PMEPRs for OFDMA systems with DSA. We consider DSA carried out over four contiguous RBs and target at a preamble sequence design which leads to low-PMEPRs for a variety of DSA schemes. There are two main reasons that we consider a chunk of four contiguous RBs: 1) In practice, four contiguous RBs may well serve most scenarios in DSA. Moreover, a large frequency band can always be divided into multiple chunks each having four RBs; 2) There exists a mathematical beauty in this setting that the maximum PMEPR can be proved to be at most 4. When more than four RBs are considered, the PMEPR upper bound may increase. However, analytical characterization of the PMEPR upper bound in this case is not straightforward. In the sequel, we introduce the related works, followed by a summary of our contributions in this work.
\subsection{Related Works}

A preamble refers to an OFDM symbol (or more) which is placed at the front of each transmission frame for a series of signal processing operations such as  synchronization and channel estimation. A preamble, known at the receiver, is desired to have time domain waveform with low PMEPR in order to avoid excessive signal distortion at the nonlinear region of  the power amplifier.

There have been numerous works on the search of waveforms with low PMEPRs. A remarkable work was done by Davis and Jedwab \cite{Davis-Jedwab-99} who have constructed $2^h$-ary Golay complementary sequences (GCSs) \cite{Golay-61} using the algebraic tool of generalized Boolean functions. In this work, the  Golay complementary pairs (GCPs) constructed by Davis and Jedwab are called the standard GCPs, and any constituent sequence in a standard GCP is called a Golay-Davis-Jadweb (GDJ) sequence. By definition, two polyphase GCSs form a GCP with zero out-of-phase aperiodic autocorrelation sums. When a GCS is spread over a contiguous group of subcarrier frequencies, each
 GCS gives rise to an OFDM waveform with a PMEPR bounded by 2 \cite{Popovic-91}. \cite{Wang-Gong-14} studied the PMEPR distribution of binary GDJ sequences.
As an alternative to GCSs, complementary (or near-complementary) sets with constituent sequences of 2 or more, have been proposed in  \cite{Schmidt-06,Schmidt-07,Yu-Gong-11,Liu-Para-Guan-14}, in which their PMEPRs are upper bounded by a small value slightly larger than 2.

GDJ sequences are excellent candidates for rapid hardware generation especially for large sequence lengths. In the context of DSA based OFDMA systems, however, the resultant PMEPR of a preamble sequence, called a subsequence in this work by taking certain sequence elements of a GDJ sequence (contiguously or non-contiguously), may be unacceptably high. \cite{CK-09} proposed hierarchical construction methods of long complementary sequences out of short ones for $2^k$-RB OFDMA systems, but they required additional resources as side information bits, which could lead to reduced spectrum efficiency. Moreover, the PMEPRs of the concatenations of any three adjacent short complementary sequences were not considered in \cite{CK-09}. In IEEE 802.11ax \cite{IEEE-ax} , there are some designs for short training fields (STF) and long
 training fields (LTF), where the widths of RBs are 26, 52, 106, 242 and 484 for different bandwidths. Although the subsequences on each RB of these training sequences have low PMEPRs, their lengths are fixed. It is found in \cite{WY-10} that certain variations of the concatenations of $2^m$ GCS are still GCSs, whose subsequences on each RB have PMEPRs upper bounded by 2. Chen proposed a construction of complementary sequences (CSs) of length $2^{m-1}+2^v$ with PMEPRs upper bounded by 4, which are subsequences on contiguous RBs of some CSs with lengths $2^m$ by deleting the last $2^{m-1}-2^v$ elements \cite{ChenCY-16}.

\subsection{Main Contributions}

This paper considers the design of preambles in OFDMA system with DSA, where an entire (contiguous) spectral band is divided into four RBs. Due to the dynamism of the spectrum allocation, each of the preambles in these four RBs may be turned on or off independently. We aim to design families of sequences, whose subsequences, all display low PMEPRs. This finding allows us to deploy a fixed preamble sequence, to an OFDMA system, which has guaranteed low-PMEPRs for any DSA schemes over the four RBs.
With the aid of GDJ sequences, we introduce two classes of preamble sequences, whose subsequences corresponding to all the DSA schemes, deployed contiguously or non-contiguously, have PMEPRs of at most 4.

\subsection{Paper Organization and Notations}

The remainder of this paper is organized as follows. Section II gives the preliminaries and the mathematical tools used in the paper. In Section III, we recall the construction of GDJ sequences, and then we discuss the PMEPR properties of subsequences of GDJ sequences. In Section IV, first, we present a class of preamble sequences whose subsequences corresponding to any number of contiguous RBs have PMEPR less than 3.3334. Then, we present a class of preamble sequences whose subsequences under contiguous DSA have PMEPR upper bounded by 4. In Section V, we study the PMEPR properties for the subsequences of the sequences in Section IV for OFDMA systems with non-contiguous DSA. Section VI compares the PMEPR properties of proposed sequences with those of $m$-sequences and Zadoff-Chu sequences by some simulation results. Section VII concludes this paper with some remarks.

We end this section by introducing some notations:
\begin{itemize}

  \item $q$ is an even integer;
  \item $\mathbb{Z}_q=\{0,1,...,q-1\}$ denotes the ring of integers modulo $q$;
  \item $\xi_q$ denotes the $q$th primitive root of unity;
  \item $|a|$ denotes the modulus of the complex number $a$;
    \item $a^*$ denotes the complex conjugation of the complex number $a$;
  \item $|\mathbf{a}|$ denotes the magnitude of the vector $\mathbf{a}$;
  \item $\mathbf{a}^T$ denotes the transposition of the vector $\mathbf{a}$;
    \item $\mathbf{a}(n:m)$ denotes the partial sequence of the sequence $\mathbf{a}$ from the $n$th element to the $m$th element;
   \item $\mathbf{a\parallel b}$ denotes the concatenation of the sequences $\mathbf{a}$ and $\mathbf{b}$;
  \item $|\mathcal{C}|$ denotes the size of the set $\mathcal{C}$;
  \item $\text{Re}(a)$ denotes the real part of the complex number $a$.
\end{itemize}
\section{Preliminaries}
\subsection{Complementary Sequence Sets}

Let $\mathbf{a}=(a(0),a(1),...,a(L-1))$ and $\mathbf{b}=(b(0),b(1),...,b(L-1))$  be complex-valued sequences of length $L$. The aperiodic cross-correlation between $\mathbf{a}$ and $\mathbf{b}$ at a time shift $\tau$ is defined as
\begin{eqnarray*}
 R_{\mathbf{a},\mathbf{b}}(\tau) &=&\left\{\begin{array}{ll}
\sum_{i=0}^{L-1-\tau}a(i) b^*(i+\tau), & 0\leq \tau\leq L-1; \\
\sum_{i=0}^{L-1+\tau}a(i-\tau) b^*(i), &-(L-1)\leq \tau\leq-1; \\
 0, & |\tau|\geq L;
   \end{array}\right.
\end{eqnarray*}
and $ R^*_{\mathbf{a},\mathbf{b}}(\tau)$ denotes the complex conjugation of $ R_{\mathbf{a},\mathbf{b}}(\tau)$. When $\mathbf{a}=\mathbf{b}$,
 $R_{\mathbf{a},\mathbf{a}}(\tau)$ is called the aperiodic auto-correlation of $\mathbf{a}$. In this case, we write  $R_{\mathbf{a},\mathbf{a}}(\tau)=R_{\mathbf{a}}(\tau)$.


\begin{Definition}[\cite{Tseng-Liu-72}]
Let ${\mathcal{A}}=(\mathbf{a}_{i})_{i=1}^{N}$ be a set of $N$ sequences of length $L$.  It is said to be a complementary sequence set (CSS) of size $N$ if $\sum_{i=1}^NR_{\mathbf{a}_{i}}(\tau)=0$ for any $\tau>0$. In this case, every $\mathbf{a}_i$ in $\mathcal{A}$ is called a complementary sequence (CS). In particular, when $N=2$, ${\mathcal{A}}$ is called a Golay complementary pair (GCP), and any constituent sequence in this pair is called a Golay complementary sequence (GCS).
\end{Definition}

\begin{Definition}
Let ${\mathcal{C}}=(\mathbf{a,b})$ be a set of two  sequences of length $L$.  It is said to be an almost complementary pair (ACP) if there exists a positive integer $\mu~(1\leq\mu\leq L-1)$ and a complex number $A\ne 0$, such that for any $\tau$, $1\le\tau<L$, we have
\begin{eqnarray*}
R_{\mathbf{a}}(\tau)+R_{\mathbf{b}}(\tau)&=&\left\{\begin{array}{ll}
                                               A, & \tau=\mu; \\
                                               0, & otherwise.
                                             \end{array}\right.
\end{eqnarray*}
Such $\mathbf{a}$ is called an almost complementary sequence (ACS).
\end{Definition}

\begin{Definition}
Let ${\mathcal{S}}=(\mathbf{a},\mathbf{b})$ and ${\mathcal{K}}=(\mathbf{c},\mathbf{d})$ be two GCPs of length $L$. $\mathcal{S}$ is said to be a Golay mate of $\mathcal{K}$ if
\begin{eqnarray*}
  R_{\mathbf{a},\mathbf{c}}(\tau)+R_{\mathbf{b},\mathbf{d}}(\tau)=0, && 0\leq\tau\leq L-1.
\end{eqnarray*}
\end{Definition}

\subsection{Generalized Boolean Functions}

For $\mathbf{x}=(x_1,x_2,\cdots,x_m)\in\mathbb{Z}_2^m$, a generalized Boolean function $f(\mathbf{x})$ is defined as a mapping $f$ from $\mathbb{Z}_2^m$ to $\mathbb{Z}_q$:
$$
f(\mathbf{x})=\sum_{i=0}^{2^m-1}a_i\prod_{k=1}^mx_k^{i_k},\,\,a_i\in \mathbb{Z}_q,
$$
where $(i_1,i_2,\cdots,i_m)$ is the binary representation of the integer $i=\sum_{k=1}^{m}2^{k-1}i_k$. For any given $f(\mathbf{x})$, we can define a sequence
\begin{eqnarray*}
\mathbf{f}&=&(f(0),f(1),\cdots, f(2^{m}-1))\\
&=&(f(0,0,\cdots,0),f(1,0,\cdots,0),\cdots,f(1,1,\cdots,1)).
\end{eqnarray*}
One can naturally associate a complex-valued sequence $\psi(\mathbf{f}^{(L)})$ of length $L$ with $\mathbf{f}^{(L)}$ as
\begin{eqnarray*}
\psi(\mathbf{f}^{(L)}) =(\xi_q^{f(0)},\xi_q^{f(1)},\cdots,\xi_q^{f(L-1)}).
\end{eqnarray*}
From now on, whenever the context is clear, we ignore the superscript of $\mathbf{f}^{(L)}$ unless the sequence length is specified.

\subsection{PMEPRs of OFDMA Symbols}

Let us consider an $L$-subcarrier OFDMA system. Without loss of generality, for $k$th transmitter, let $\Omega$ be the subcarrier index set allocated to this system. For a $\mathbb{Z}_q$-valued sequence $\mathbf{a}=(a(0),a(1),...,a(L-1))$, there is a sequence $\mathbf{\tilde{a}}=(\tilde{a}(0),\tilde{a}(1),...,\tilde{a}(L-1))$ corresponding to $\Omega$, where
\begin{eqnarray*}
 \tilde{a}(i) &=&\left\{\begin{array}{ll}
\xi_q^{a(i)}, & i\in\Omega, \\
0, &i\not\in\Omega. \\
   \end{array}\right.
\end{eqnarray*}
The transmitted OFDMA signal is the real part of the complex envelope, which can be written as
\begin{eqnarray*}
S_\mathbf{\tilde{a}}(t)= \sum_{i=0}^{L-1}\tilde{a}(i)e^{2\pi\left(f_{c}+i{\Delta}f\right)t\sqrt{-1}},~ 0\leq t<T,
\end{eqnarray*}
where $f_{c}$ denotes the carrier frequency and ${\Delta}f=\frac{1}{T}$ denotes the subcarrier spacing, with $T$ being the OFDMA symbol duration. The sequence $\mathbf{\tilde{a}}$ of length $L$ is called the modulating codeword of the OFDMA symbol for the subcarrier set $\Omega$.

The instantaneous power of an OFDMA sequence (codeword) ${\mathbf{\tilde{a}}}$ is given by
\begin{eqnarray}\label{eq-p}
P_\mathbf{\tilde{a}}(t)=R_\mathbf{\tilde{a}}(0)+2\text{Re}\left(
\sum_{\tau=1}^{L-1}R_\mathbf{\tilde{a}}(\tau)
e^{2\pi(\tau{\Delta}f)t\sqrt{-1}}\right).
\end{eqnarray}
The peak-to-mean energy power ratio (PMEPR) of the OFDMA sequence ${\mathbf{\tilde{a}}}$ is then defined as:
\begin{eqnarray}\label{eq-pmepr}
\hbox{PMEPR}(\mathbf{\tilde{a}}) =\frac{\underset{t\in[0,T)}{\sup} P_{\mathbf{\tilde{a}}}(t)}{P_{av}({\mathbf{\tilde{a}}})},
\end{eqnarray}
where $P_{av}({\mathbf{\tilde{a}}})$ is the average power of $\mathbf{\tilde{a}}$, and
 \begin{eqnarray}\label{eq-pav}
 P_{av}({\mathbf{\tilde{a}}})=\frac{1}{T}\int_{[0,T]}P_{\mathbf{\tilde{a}}}(t)dt=
 \parallel\mathbf{\tilde{a}}\parallel^2=R_{\mathbf{\tilde{a}}}(0).
 \end{eqnarray}
 In addition, define the ``instantaneous-to-mean envelope power
ratio (IMEPR)'' of $\mathbf{\tilde{a}}$ as
\begin{eqnarray*}
\hbox{IMEPR}(\mathbf{\tilde{a}},t)=P_{\mathbf{\tilde{a}}}(t)/P_{av}({\mathbf{\tilde{a}}}),
\end{eqnarray*}
clearly, $\hbox{PMEPR}(\mathbf{\tilde{a}})=sup_{t\in[0,T)}\hbox{IMEPR}(\mathbf{\tilde{a}},t).$
 Accordingly, the PMEPR of a sequence set ${\mathcal{A}}=\{\mathbf{a}_1,\mathbf{a}_2,\cdots,
\mathbf{a}_N\}$ is  defined as
\begin{eqnarray*}
\hbox{PMEPR}({\mathcal{A}}) &=& \underset{\mathbf{a}_i\in {\mathcal{A}}}{\max}~\hbox{PMEPR}(\mathbf{a}_i).
\end{eqnarray*}
Combining Eqs. (\ref{eq-p}) and (\ref{eq-pmepr}), it can be derived that for sequence set ${\mathcal{A}}$ of length $n$, we have
\begin{eqnarray}\label{eq-pmeprN}
\emph{PMEPR}(\mathbf{a}_i)\leq \frac{1}{R_{\mathbf{a}_i}(0)}\left(\sum_{k=1}^NR_{\mathbf{a}_k}(0)
+2\sum_{\tau=1}^{L-1}\left|\sum_{k=1}^NR_{\mathbf{a}_k}(\tau)\right|\right)&&(1\leq i\leq N),
\end{eqnarray}
which implies the following lemma:
\begin{Lemma}[\cite{Liu-Guan-16}]\label{lem-bound}
Let $\mathcal{A}$ be a CSS of size $N$ in which all the sequences have the same energy. Then the PMEPR of $\mathcal{A}$ is upper bounded by $N$.
\end{Lemma}

Lemma \ref{lem-bound} is useful to evaluate the PMEPR of a sequence in the sequel.

\section{PMEPR Properties of the Subsequences of GDJ sequences}

In this section, first, we give some notations needed in the sequel. Then, we introduce the GDJ sequences and discuss the PMEPR properties of their subsequences.

For a $\mathbb{Z}_q$-valued sequence $\mathbf{a}$ of length $L=4H$, where $H$ is a positive integer which can be seen as the width of a RB, define four subsequences corresponding to the RBs from it as
\begin{eqnarray*}
  \mathbf{a}_1=(a_0,a_1,...,a_{H-1}),&&\mathbf{a}_2=(a_H,a_{H+1},...,a_{2H-1}),\\ \mathbf{a}_3=(a_{2H},a_{2H+1},...,a_{3H-1}),&&\mathbf{a}_4=(a_{3H},a_{3H+1},...,a_{4H-1}).
\end{eqnarray*}
Based on these subsequences and a zero-sequence $\mathbf{0}_H=(0,...,0)$, we define 15 subsequences as shown in Fig. \ref{As},
\begin{figure}[htbp]
\centering
\includegraphics[width=5.25in]{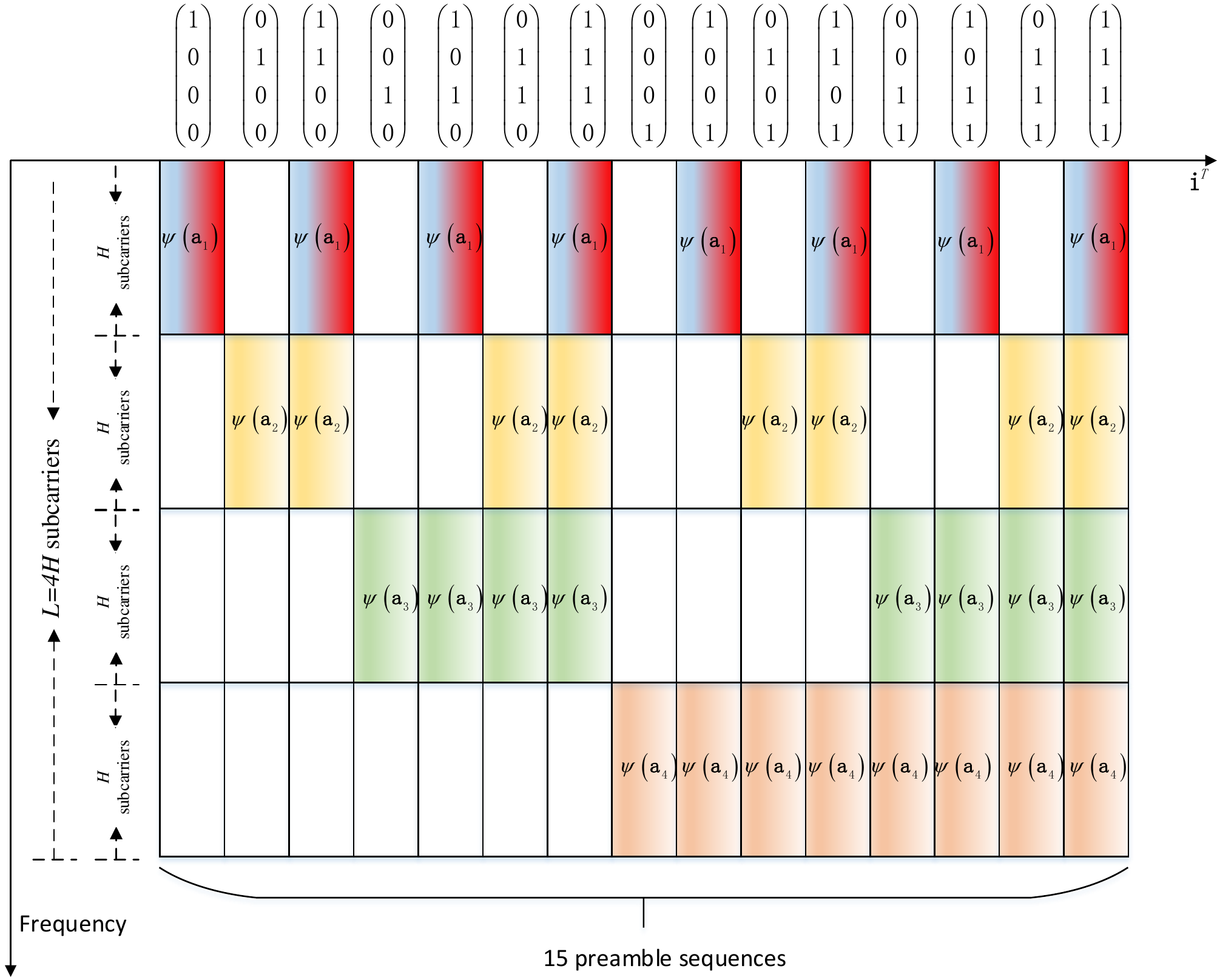}\\
\caption{Subsequences $\mathbf{A}_i~(1\leq i\leq 15)$ constructed from $\mathbf{a}$ and $\mathbf{0}_H$}\label{As}
\end{figure}
where $L$ is the number of subcarriers, $i=\sum_{k=1}^4i_k2^{k-1},~ \mathbf{A}_{i}=\mathbf{A}_{(i_1,i_2,i_3,i_4)}=(i_1\psi(\mathbf{a}_1)\parallel
 i_2\psi(\mathbf{a}_2)\parallel i_3\psi(\mathbf{a}_3)\parallel$ $i_4\psi(\mathbf{a}_4))$, $\mathbf{i}=(i_1,i_2,i_3,i_4)$ and $i_k=1~(1\leq k\leq4)$ means the $k$th RB has been allocated to the transmitter, otherwise, $i_k=0$. It is straightforward that $\mathbf{A}_{15}=\mathbf{A}_{(1,1,1,1)}=\psi(\mathbf{a})$. These sequences can be divided into 2 disjoint subsequence sets of $\mathbf{a}$: contiguous subsequence set $\mathcal{C}=\{\mathbf{A}_1,\mathbf{A}_2,\mathbf{A}_3,\mathbf{A}_4,\mathbf{A}_6,\mathbf{A}_7,
 \mathbf{A}_8,\mathbf{A}_{12},\mathbf{A}_{14},\mathbf{A}_{15}\}$ and non-contiguous subsequence set $\mathcal{NC}=\{\mathbf{A}_5,\mathbf{A}_9,\mathbf{A}_{10},\mathbf{A}_{11},$ $\mathbf{A}_{13}\}$.
 (Note that in IEEE 802.11ax \cite{IEEE-ax}, they only considered subsequences $\mathbf{A}_1,\mathbf{A}_2,\mathbf{A}_3,\mathbf{A}_4,\mathbf{A}_8,\mathbf{A}_{12},\mathbf{A}_{15}$,
 which are included in $\mathcal{C}$.)

In this paper, we define $\hbox{PMEPR}_C(\mathbf{a})$ as the maximum PMEPR of $\mathcal{C}$, i.e., $$\hbox{PMEPR}_{\mathcal{C}}(\mathbf{a})=\max_{s\in\{1,2,3,4,6,7,8,12,14,15\}}\hbox{PMEPR}(\mathbf{A}_s),$$
and we define $\hbox{PMEPR}_{\mathcal{NC}}(\mathbf{a})$ as the maximum PMEPR of $\mathcal{NC}$, i.e., $$\hbox{PMEPR}_{\mathcal{NC}}(\mathbf{a})=\max_{s\in\{5,9,10,11,13\}}\hbox{PMEPR}(\mathbf{A}_s).$$
In addition, we define $\hbox{PMEPR}_{\mathcal{A}}(\mathbf{a})$ as the maximum PMEPR of $\mathcal{A}=\mathcal{NC}\bigcup\mathcal{C}$, i.e.,
$$\hbox{PMEPR}_{\mathcal{A}}(\mathbf{a})=\max_{s\in\{1,2,...,15\}}\hbox{PMEPR}(\mathbf{A}_s)
=max\{\hbox{PMEPR}_{\mathcal{C}}(\mathbf{a}),\hbox{PMEPR}_{\mathcal{NC}}(\mathbf{a})\}.$$

The following lemma gives the construction of GDJ sequences.

\begin{Lemma}[Corollary 11 of \cite{Paterson-00}, Theorem 3.3 of \cite{R.A-08}] \label{lem-GCP}
Let $q$ be an even integer and $m$ be a positive integer. Let
\begin{eqnarray*}
a(\mathbf{x}) &=& \frac{q}{2}\sum_{k=1}^{m-1}x_{\pi(k)}x_{\pi(k+1)}+\sum_{k=1}^mc_kx_k+c, \\
b(\mathbf{x}) &=& a(\mathbf{x})+ \frac{q}{2}x_{\pi(1)},\\
 c(\mathbf{x}) &=& a(\mathbf{x})+ \frac{q}{2}x_{\pi(m)},
\end{eqnarray*}
where $\pi$ is a permutation of $\{1,2,\cdots,m\}$ and $\mathbf{x}\in\mathbb{Z}_2^m,c_k,c\in\mathbb{Z}_q$. Then $(\psi(\mathbf{a}),\psi(\mathbf{b})))$ and $(\psi(\mathbf{a}),\psi(\mathbf{c}))$ are GCPs of length $2^m$. In particular, let
\begin{eqnarray*}
d(\mathbf{x}) &=& a(\mathbf{x})+ \frac{q}{2}x_{\pi(1)}+ \frac{q}{2}x_{\pi(m)},
\end{eqnarray*}
then $(\psi(\mathbf{a}),\psi(\mathbf{c}))$ is the Golay mate of $(\psi(\mathbf{b}),\psi(\mathbf{d}))$.
\end{Lemma}

\begin{Remark}
In particular, when $q=2^s$ where $s$ is a positive integer, Lemma \ref{lem-GCP} is the Theorem 3 of \cite{Davis-Jedwab-99}. Any sequence constructed by Lemma \ref{lem-GCP} is called a Golay-Davis-Jedwab (GDJ) sequence.
\end{Remark}

\begin{Remark}\label{re-gdjpmepr}
The PMEPR of every GDJ sequence is upper bounded by 2.
\end{Remark}


The PMEPRs of the subsequences of GDJ sequences may be large. We illustrate this by the following example.

\begin{Example}\label{eg-ab}
For $q=2,m=9$, let $\pi$ be a permutation of $\{1,2,...,9\}$ with $\pi(1)=7,\pi(2)=9,\pi(3)=6,\pi(4)=3,\pi(5)=1,\pi(6)=5,\pi(7)=4,\pi(8)=8,\pi(9)=2$, and $a(\mathbf{x})=\sum_{k=1}^{m-1}x_{\pi(k)}x_{\pi(k+1)}$. Tables \ref{table-cgdj} and  \ref{table-ncgdj} show the PMEPRs for the subsequences of $\mathbf{a}$. It can be observed that $\hbox{PMEPR}_\mathcal{C}(\mathbf{a})=8$ and $\hbox{PMEPR}_{\mathcal{NC}}(\mathbf{a})=4$, which would be undesirable in practical communication systems.
\begin{table}[htbp]
  \centering
    \caption{The PMEPRs of the contiguous subsequences of the GDJ sequence $\mathbf{a}$ in Example \ref{eg-ab}}\label{table-cgdj}
    {\small
  \begin{tabular}{|c|c|c|c|c|c|c|c|c|c|c|}
  \hline
  $\mathbf{A}_s$ &$\mathbf{A}_1$ & $\mathbf{A}_2$ & $\mathbf{A}_3$ & $\mathbf{A}_4$ & $\mathbf{A}_6$ & $\mathbf{A}_7$ & $\mathbf{A}_8$ & $\mathbf{A}_{12}$&$\mathbf{A}_{14}$&$\mathbf{A}_{15}$\\\hline
  PMEPR$(\mathbf{A}_s)$&8.0000 & 7.6216 &4.0000 & 7.3279 & 3.8469 &3.0682&7.7605&3.9155&2.9077&2.0000\\
  \hline
\end{tabular}}

\end{table}

\begin{table}[htbp]
  \centering
    \caption{The PMEPRs of the non-contiguous subsequences of the GDJ sequence $\mathbf{a}$ in Example \ref{eg-ab}}\label{table-ncgdj}
  \begin{tabular}{|c|c|c|c|c|c|}
  \hline
  $\mathbf{A}_s$ &$\mathbf{A}_5$ & $\mathbf{A}_9$ & $\mathbf{A}_{10}$ & $\mathbf{A}_{11}$ & $\mathbf{A}_{13}$ \\\hline
 PMEPR$(\mathbf{A}_s)$& 4.0000 &4.0000  & 3.9215 &  3.1025&3.3668\\
  \hline
\end{tabular}

\end{table}

\end{Example}

Example \ref{eg-ab} motivates us to search some GDJ sequences whose subsequences have low-PMEPR properties.
\section{Proposed Low-PMEPR Preamble Sequences for Contiguous Spectrum Allocation}

In this section, we introduce two classes of preamble sequences with good PMEPR properties of the subsequences for the OFDMA system of contiguous frequency bands allocation.
\subsection{Preamble Sequences with PMEPRs of Contiguous Subsequences Upper Bounded by $\frac{10}{3}$}

\begin{Theorem}\label{thm-acp1}
Let $m\geq 2$ be a positive integer and $q$ be an even integer. Let
\begin{eqnarray*}
a(\mathbf{x}) &=& \frac{q}{2}\sum_{k=1}^{m-1}x_{\pi(k)}x_{\pi(k+1)}+\sum_{k=1}^mc_kx_k+c, \\
b(\mathbf{x}) &=& a(\mathbf{x})+ \frac{q}{2}x_{\pi(1)},\\
\end{eqnarray*}
where $\mathbf{x}\in\mathbb{Z}_2^m,c_k,c\in\mathbb{Z}_q$. Let $H=2^{m-2}$ and  $\pi$ is a permutation of $\{1,2,\cdots,m\}$ satisfying $\pi(m)=m$ and $\pi(m-1)=m-1$, then the GDJ sequences have the following properties:
\begin{enumerate}
\item for $s=7,14$, $(\mathbf{A}_s,\mathbf{B}_s)$ is an ACP, i.e.,
\begin{eqnarray*}
R_{\mathbf{A}_s}(\tau)+R_{\mathbf{B}_s}(\tau)&=&
\left\{\begin{array}{ll}
6H, & \tau=0,\\
\xi_q^{-{q\over 2}\left({s\over 7}-1\right)-c_m}\cdot 2H, &\tau=2H,\\
0, & \text{otherwise};
\end{array}\right.
\end{eqnarray*}
\item for $s=3,6,12$, $(\mathbf{A}_s,\mathbf{B}_s)$ is a GCP;
\item for $s=1,2,4,8$, $(\mathbf{A}_s,\mathbf{B}_s)$ is a GCP.
\end{enumerate}

\end{Theorem}

\begin{proof}
\begin{enumerate}
\item
For $s=7,14$, when $0\leq i\leq L-1-\tau$, let $j=i+\tau$, then
\begin{eqnarray*}
A_s(i)-A_s(j)&=&{q\over 2}\sum_{k=1}^{m-1}(i_{\pi(k)}       i_{\pi(k+1)}-j_{\pi(k)}j_{\pi(k+1)})+\sum_{k=1}^{m}c_k(i_k-j_k),\\
B_s(i)-B_s(j)&=&A_s(i)
-A_s(j)
+\frac{q}{2}(i_{\pi(1)}-j_{\pi(1)}).
\end{eqnarray*}
Hence we have
\begin{eqnarray}\label{eq-3}
R_{\mathbf{A}_s}(\tau)+R_{\mathbf{B}_s}(\tau)
&=&\sum_{i=0}^{L-1-\tau}\left[\xi_q^{A_s(i)-A_s(j)}
+\xi_q^{B_s(i)-B_s(j)}\right]\nonumber\\
&=&\sum_{i=0}^{L-1-\tau}\xi_q^{A_s(i)-
A_s(j)}\left[1+(-1)^{i_{\pi(1)}-j_{\pi(1)}}\right]\nonumber\\
&=&2\sum_{i\in J(\tau)}\xi_q^{A_s(i)-A_s(j)},
\end{eqnarray}
where $J(\tau)=\{0\le i\le L-1-\tau: i_{\pi(1)}=j_{\pi(1)}\}$.

\begin{itemize}
\item
When $\tau=2^{m-1}$, since $j=i+2^{m-1}$, $\pi(m)=m$ and $\pi(m-1)=m-1$, it can be obtained that $j_{m}=1,i_m=0$, $i_{m-1}=j_{m-1}=s$ and $i_{\pi(t)}=j_{\pi(t)}$ for $1\leq t\leq m-2$, which results in $A_s(i)-A_s(j)=-{q\over 2}({s\over 7}-1)-c_m$ and $J(\tau)=\{0,1,\cdots,2^{m-2}-1\}$. Hence (\ref{eq-3}) can be reduced as
\begin{eqnarray*}
R_{\mathbf{A}_s}(\tau)+R_{\mathbf{B}_s}(\tau)
=\sum_{i=0}^{2^{m-2}-1}2\xi_q^{-{q\over 2}\left({s\over 7}-1\right)-c_m}=2^{m-1}\cdot \xi_q^{-{q\over 2}\left({s\over 7}-1\right)-c_m}.
\end{eqnarray*}

\item
When $\tau\in\{1,2,\cdots,L-1\}\setminus\{2^{m-1}\}$, for any $i\in J(\tau)$, let $t$ be the smallest integer in $\{1,2,...,m\}$ such that $i_{\pi(t)}\ne j_{\pi(t)}$, which implies $2\leq t< m$.
Let $i'$ and $j'$ be integers which are different from $i$ and $j$ in only one position $\pi(t-1)$, i.e., $i'_{\pi(t-1)}=1-i_{\pi(t-1)}$, respectively, and so $j'=i'+\tau$. Then we have
\begin{eqnarray*}
A_s(i)-A_s(j)-\left(A_s(i')
-A_s(j')\right) \equiv\frac{q}{2}~~(\bmod~q),
\end{eqnarray*}
which implies
$\xi_q^{A_s(i)-A_s(j)}=-\xi_q^{A_s(i')-A_s(j')}$. Hence (\ref{eq-3}) can be reduced to
\begin{eqnarray*}
R_{\mathbf{A}_s}(\tau)+R_{\mathbf{B}_s}(\tau)&=&
\sum_{i\in J(\tau)}\xi_q^{A_s(i)-A_s(j)}
+\sum_{i'\in J(\tau)}\xi_q^{A_s(i')-A_s(j')}=0.
\end{eqnarray*}
\end{itemize}
Combining these two cases, we have
\begin{eqnarray*}
R_{\mathbf{A}_s}(\tau)+R_{\mathbf{B}_s}(\tau)&=&
\left\{\begin{array}{ll}
\xi_q^{-{q\over 2}\left({s\over 7}-1\right)-c_m}\cdot 2H, &\tau=2H,\\
0, & \text{otherwise};
\end{array}\right.
\end{eqnarray*}
i.e., $\left(\mathbf{A}_s,\mathbf{B}_s\right)$ is an ACP.

\item
\begin{itemize}
\item
For $s=3$, and $0<\tau<2H-1$, we have $i_m=0$. Then
\begin{eqnarray*}
A_3(i) &=& \frac{q}{2}\sum_{k=1}^{m-2}i_{\pi(k)}i_{\pi(k+1)}+\sum_{k=1}^{m-1}c_ki_k+c, \\
B_3(i) &=& A_3(i)+ \frac{q}{2}i_{\pi(1)}.
\end{eqnarray*}
According to Lemma \ref{lem-GCP}, $(\mathbf{A}_3,\mathbf{B}_3)$ is a GCP of length $2H$.
%

\item
For $s=6,~0<\tau<2H-1$, let $j=i+\tau$ and  $H\leq i,j<3H$, we have $i_{m-1}=1,i_m=0$ or $i_{m-1}=0,i_m=1$. Then,
\begin{eqnarray}\label{eq-2-1}
A_6(i)-A_6(j)&=&
\frac{q}{2}\sum_{k=1}^{m-2}\left(i_{\pi(k)}i_{\pi(k+1)}-j_{\pi(k)}j_{\pi(k+1)}\right)+\sum_{k=1}^{m}c_k(i_k-j_k), \\
\label{eq-2-2}
B_6(i)-B_6(j)&=& A_6(i)-A_6(j)+ \frac{q}{2}\left(i_{\pi(1)}-j_{\pi(1)}\right).
\end{eqnarray}
If $i_{\pi(1)}\ne j_{\pi(1)}$, we have
$$\xi_q^{A_6(i)-A_6(j)}=-\xi_q^{B_6(i)-B_6(j)}.$$
If $i_{\pi(1)}=j_{\pi(1)}$, let $t$ be the smallest integer such that $i_{\pi(t)}\ne j_{\pi(t)}$.  Let $i'$ and $j'$ be integers which are different from $i$ and $j$ in only one position $\pi(t-1)$, i.e., $i'_{\pi(t-1)}=1-i_{\pi(t-1)}$, respectively, and so $j'=i'+\tau$. Hence we have
\begin{eqnarray*}
\left(B_6(i)-B_6(j)\right)-\left(B_6(i')-B_6(j')\right)&=&
\left(A_6(i)-A_6(j)\right)-\left(A_6(i')-A_6(j')\right)\\
&\equiv&\frac{q}{2}
~~(\bmod~q),
\end{eqnarray*}
and then
\begin{eqnarray*}
\sum_{i=H}^{3H-1-u}\left(\xi_q^{A_6(i)-A_6(j)}+\xi_q^{B_6(i)-B_6(j)}\right)=0.
\end{eqnarray*}

The proof for $s=12$ is similar to the case for $s=3$, and hence we omit it here.
\end{itemize}

\item
For $0\leq i<H=2^{m-2}$, we have $i_{m-1}=0$ and $i_m=0$. Then,
\begin{eqnarray*}
A_1(i) &=& \frac{q}{2}\sum_{k=1}^{m-3}i_{\pi(k)}i_{\pi(k+1)}
+\sum_{k=1}^{m-2}c_ki_k+c, \\
B_1(i) &=& A_1(i)+ \frac{q}{2}i_{\pi(1)}.
\end{eqnarray*}
According to Lemma \ref{lem-GCP}, $(\mathbf{A}_1,\mathbf{B}_1)$ is a GCP.

The proof of $s=2,4,8$ is similar with that of $s=1$, we omit it here.

\end{enumerate}
\end{proof}

\begin{Remark}
Actually, when $q=2$, $(\mathbf{A}_7, \mathbf{B}_{7})$ and $(\mathbf{A}_{14},\mathbf{B}_{14})$ are two binary Z-complementary pairs given in \cite{ChenCY-SPL-17}.
\end{Remark}

According to Lemma \ref{lem-bound}, Remark \ref{re-gdjpmepr} and Eq. (\ref{eq-pmeprN}), we can get an upper bound on the PMEPR of $\mathbf{A}_s$ in Theorem \ref{thm-acp1}.
\begin{Corollary}
With the same notations as Theorem \ref{thm-acp1}, we have
\begin{eqnarray*}
\hbox{PMEPR}(\mathbf{A}_s)&\leq&
\left\{\begin{array}{ll}
\frac{10}{3}, & s=7,14,\\
2, &s=1,2,3,4,6,8,12,15.\\
\end{array}\right.
\end{eqnarray*}
It is clear that
\begin{eqnarray*}
\hbox{PMEPR}_\mathcal{C}(\mathbf{a})\leq \frac{10}{3}.
\end{eqnarray*}
\end{Corollary}

\begin{Example}\label{eg-1}
For $q=2$, let $\pi$ be the identical permutation of $\{1,2,...,m\}$, and $a(\mathbf{x})=\sum_{k=1}^{m-1}x_{k}x_{k+1}$. Table \ref{table-eg-1} shows the PMEPR properties of the contiguous subsequences of $\mathbf{a}$ for various $s$ and $m$. It can be observed that $\hbox{PMEPR}_\mathcal{C}(\mathbf{a})\leq\frac{10}{3}$.

\begin{table}[htbp]
  \centering
    \caption{PMEPR comparison for the contiguous subsequences of $\mathbf{a}$ in Example \ref{eg-1} for $m=3,4,5,6$}\label{table-eg-1}
  {\small{\begin{tabular}{|c|c|c|c|c|c|c|c|c|c|c|}
  \hline
 \diagbox{$L=2^m$}{PMEPR}{$\mathbf{A}_s$}  &$\mathbf{A}_1$ & $\mathbf{A}_2$ & $\mathbf{A}_3$ & $\mathbf{A}_4$ & $\mathbf{A}_6$ & $\mathbf{A}_7$ & $\mathbf{A}_8$ & $\mathbf{A}_{12}$&$\mathbf{A}_{14}$&$\mathbf{A}_{15}$\\\hline
 8&2.0000 &2.0000 &1.7071 & 2.0000 & 1.7071 &2.6667&2.0000&1.7071&3.3166& 2.0000\\
  \hline
  16&1.7071 & 1.7071 &2.0000 & 1.7071 &2.0000 &3.0000&1.7071&2.0000&1.8844& 1.7071\\
  \hline
   32&2.0000 & 2.0000 &1.8210 &2.0000 & 1.8210&3.1910&2.0000&1.8210&3.3274& 2.0000\\
  \hline
   64&1.8210 & 1.8210 &2.0000 & 1.8210 & 2.0000 &3.1910&1.8210&2.0000&2.9419&1.8210\\
  \hline
\end{tabular}
}}
\end{table}

\end{Example}


\subsection{Preamble Sequences with PMEPRs of Contiguous Subsequences Upper Bounded by  4}
\begin{Theorem}\label{thm-gcs1}
Let $m\geq 2$ be a positive integer and $q$ be an even integer. Let
\begin{eqnarray*}
a(\mathbf{x}) &=& \frac{q}{2}\sum_{k=1}^{m-1}x_{\pi(k)}x_{\pi(k+1)}+\sum_{k=1}^mc_kx_k+c, \\
b(\mathbf{x}) &=& a(\mathbf{x})+ \frac{q}{2}x_{\pi(1)},\\
 d(\mathbf{x}) &=& a(\mathbf{x})+ \frac{q}{2}x_{m-1}, \\
e(\mathbf{x}) &=& a(\mathbf{x})+ \frac{q}{2}x_{\pi(1)}+ \frac{q}{2}x_{m-1}.
\end{eqnarray*}
where $\mathbf{x}\in\mathbb{Z}_2^m,c_k,c\in\mathbb{Z}_q$. Let $H=2^{m-2}$ and  $\pi$ is a permutation of $\{1,2,\cdots,m\}$ satisfying $\pi(m)=m-1$ and $\pi(m-1)=m$,
then the GDJ sequences have the following properties:
\begin{enumerate}
\item for $s=7,14$, $(\mathbf{A}_s,\mathbf{B}_s)$ is an ACP, i.e.,
\begin{eqnarray*}
R_{\mathbf{A}_s}(\tau)+R_{\mathbf{B}_s}(\tau)&=&
\left\{\begin{array}{ll}
6H,&\tau=0,\\
\xi_q^{-{q\over 2}\left({s\over 7}-1\right)-c_{m-1}}\cdot 2H, &\tau=H,\\
0, & \text{otherwise};
\end{array}\right.
\end{eqnarray*}
\item  for $s=3,12$, $(\mathbf{A}_s,\mathbf{B}_s,\mathbf{D}_s,\mathbf{E}_s)$ is a CSS;
\item   for $s=6$, $(\mathbf{A}_s,\mathbf{B}_s)$ is a GCP;
\item   for $s=1,2,4,8$, $(\mathbf{A}_s,\mathbf{B}_s)$ is a GCP.
\end{enumerate}
\end{Theorem}

\begin{proof}
Let $A_s(\mathbf{x})$ $(s=1,2,3,4,6,7,8,12,14)$ be the generalized Boolean function corresponding to the sequence $\mathbf{A}_s$.

\begin{enumerate}
\item
For $s=7,14$, when $0\leq i\leq L-1-\tau$, let $j=i+\tau$, then
\begin{eqnarray*}
A_s(i)-A_s(j)
&=&{q\over 2}\sum_{k=1}^{m-1}\left(i_{\pi(k)}       i_{\pi(k+1)}-j_{\pi(k)}j_{\pi(k+1)}\right)+\sum_{k=1}^{m}c_k\left(i_k-j_k\right),\\
B_s(i)-B_s(j)&=&A_s(i)
-A_s(j)
+\frac{q}{2}\left(i_{\pi(1)}-j_{\pi(1)}\right).
\end{eqnarray*}
Hence we have
\begin{eqnarray}\label{eq-a1b1}
R_{\mathbf{A}_s}(\tau)+R_{\mathbf{B}_s}(\tau)
&=&\sum_{i=0}^{L-1-\tau}\left[\xi_q^{A_s(i)-
A_s(j)}+\xi_q^{B_s(i)-B_s(j)}\right]\nonumber\\
&=&\sum_{i=0}^{L-1-\tau}\xi_q^{A_s(i)-
A_s(j)}\left[1+(-1)^{i_{\pi(1)}-j_{\pi(1)}}\right].
\end{eqnarray}

\begin{itemize}
\item
When $\tau=2^{m-2}$, since $j=i+2^{m-2}$, $\pi(m)=m-1$ and $\pi(m-1)=m$, it can be obtained that
\begin{eqnarray*}
\left\{\begin{array}{ll}
i_{m}=0,i_{m-1}=0,j_m=0,j_{m-1}=1, & 0\le i\le 2^{m-2}-1,\\
i_{m}=0,i_{m-1}=1,j_m=1,j_{m-1}=0, & 2^{m-2}\le i\le 2^{m-1}-1,\\
i_{m}=1,i_{m-1}=0,j_m=1,j_{m-1}=1, & 2^{m-1}\le i\le 2^{m-1}+2^{m-2}-1,\\
\end{array}\right.
\end{eqnarray*}
and $i_{\pi(t)}=j_{\pi(t)}$ for $1\leq t\leq m-2$, which leads to \begin{eqnarray*}
a(i)-a(j)&=&
\left\{\begin{array}{ll}
-c_{m-1}, & 0\le i\le 2^{m-2}-1,\\
-{q\over 2}j_{\pi(m-2)}-c_m+c_{m-1}, & 2^{m-2}\le i\le 2^{m-1}-1,\\
-{q\over 2}-c_{m-1}, & 2^{m-1}\le i\le 2^{m-1}+2^{m-2}-1,\\
\end{array}\right.
\end{eqnarray*}
where $A_s(i)=a(i+H\cdot s)$.
Hence Eq. (\ref{eq-a1b1}) can be reduced to
\begin{eqnarray*}
R_{\mathbf{A}_s}(\tau)+R_{\mathbf{B}_s}(\tau)
&=&2\sum_{i=0}^{2^{m-1}-1}\xi_q^{A_s(i)-A_s(j)}\\
&=&2\sum_{i=0}^{2^{m-2}-1}\xi_q^{A_s(i)-A_s(j)}+2\sum_{i=2^{m-2}}^{2^{m-1}-1}
\xi_q^{A_s(i)-A_s(j)}\\
&=&2H\cdot \xi_q^{-{q\over 2}\left({s\over 7}-1\right)-c_{m-1}}.
\end{eqnarray*}

\item When $\tau\in\{1,2,\cdots,L-1\}\setminus\{2^{m-2}\},0\leq i,j\leq L-1$, Eq. (\ref{eq-a1b1}) is equal to
\begin{eqnarray}\label{eq-a1b1-1}
R_{\mathbf{A}_7}(\tau)+R_{\mathbf{B}_7}(\tau)&=&2\sum_{i\in J(\tau)}\xi_q^{A_7(i)-A_7(j)},
\end{eqnarray}
where $J(\tau)=\{0\le i\le L-1-\tau: i_{\pi(1)}=j_{\pi(1)}\}$. For any $i\in J(\tau)$, let $t$ be the smallest integer such that $i_{\pi(t)}\ne j_{\pi(t)}$. which implies $2\leq t\leq m$.
Let $i'$ and $j'$ be integers which are different from $i$ and $j$ in only one position $\pi(t-1)$, i.e., $i'_{\pi(t-1)}=1-i_{\pi(t-1)}$, respectively, and so $j'=i'+\tau$. Then we have
\begin{eqnarray*}
A_7(i)-A_7(j)-(A_7(i')-A_7(j')) \equiv\frac{q}{2}~~(\bmod~q),
\end{eqnarray*}
which implies
$\xi_q^{A_7(i)-A_7(j)}=(-1)\xi_q^{A_7(i')-A_7(j')}$. Hence (\ref{eq-a1b1-1}) is equal to
\begin{eqnarray*}
R_{\mathbf{A}_7}(\tau)+R_{\mathbf{B}_7}(\tau)&=&\sum_{i\in J(\tau)}\xi_q^{A_7(i)-A_7(j)}
+\sum_{i'\in J(\tau)}\xi_q^{A_7(i')-A_7(j')}=0.
\end{eqnarray*}
In a similar way, one can prove that $(\mathbf{A}_{14},\mathbf{B}_{14})$ is an ACP.
\end{itemize}
Combining these two cases, we have the result of 1).
\item
According to Theorem 12 of \cite{Paterson-00}, it can be easily obtained that $(\mathbf{A}_s,\mathbf{B}_s,\mathbf{D}_s,\mathbf{E}_s)$ is a CSS for $s=3,12$, which implies the result of 2).
\item
For $0<\tau\leq2H-1$, let $j=i+\tau$ and  $H\leq i,j<3H$, we have $i_{m-1}=1,i_m=0$ or $i_{m-1}=0,i_m=1$. Then, we have
\begin{eqnarray}\label{eq-3-1}
        A_6(i)-A_6(j) &=& \frac{q}{2}\sum_{k=1}^{m-2}\left(i_{\pi(k)}
        i_{\pi(k+1)}-j_{\pi(k)}j_{\pi(k+1)}\right)+
        \sum_{k=1}^{m}c_k\left(i_k-j_k\right), \\
        \label{eq-3-2}
       B_6(i)-B_6(j) &=& A_6(i)-A_6(j)+\frac{q}{2}\left(i_{\pi(1)}-j_{\pi(1)}\right).
      \end{eqnarray}
If $i_{\pi(1)}\ne j_{\pi(1)}$, then
$$\xi_q^{A_6(i)-A_6(j)}=(-1)\xi_q^{B_6(i)-B_6(j)}. $$

If $i_{\pi(1)}=j_{\pi(1)}$ and $\tau=H$, which implies $H\leq i<2H$ and $i_{m-1}=1,i_m=0,j_{m-1}=0,j_m=1$, then
$$B_6(i)-B_6(j)=A_6(i)-A_6(j)=-\frac{q}{2}j_{\pi(m-2)}+c_{m-1}-c_m,$$
and then
$$\sum_{i=H}^{3H-1-\tau}\xi_q^{B_6(i)-B_6(j)}=
\sum_{i=H}^{3H-1-\tau}\xi_q^{A_6(i)-A_6(j)}=
\sum_{i=H}^{2H-1}\xi_q^{-\frac{q}{2}j_{\pi(m-2)}+c_{m-1}-c_m}=0.$$

If $i_{\pi(1)}=j_{\pi(1)}$ and $\tau\ne H$, let $t$ be the smallest integer such that $i_{\pi(t)}\ne j_{\pi(t)}$.  Let $i'$ and $j'$ be integers which are different from $i$ and $j$ in only one position $\pi(t-1)$, i.e., $i'_{\pi(t-1)}=1-i_{\pi(t-1)}$, respectively, and so $j'=i'+\tau$. Then, we have $t\leq m-1$. Otherwise, since $i_{\pi(k)}=j_{\pi(k)}$ for $k\in\{1,2,...,m-1\}$ which implies that $j=i+2^{m-2},\tau=2^{m-2}$, it contradicts the assumption. Therefore, we have
      $H\leq i',j'\leq 3H-1$. According to (\ref{eq-3-1}) and (\ref{eq-3-2}), we have
      \begin{eqnarray*}
     \left(B_6(i)-B_6(j)\right)-\left(B_6(i')-B_6(j')\right)=\left(A_6(i)-A_6(j)\right)
     -\left(A_6(i')-A_6(j')\right)\equiv\frac{q}{2}
     ~~(\bmod~q).
      \end{eqnarray*}
and then
      \begin{eqnarray*}
        \sum_{i=H}^{3H-1-\tau}\left(\xi_q^{A_6(i)-A_6(j)}+
        \xi_q^{B_6(i)-B_6(j)}\right)=0.
      \end{eqnarray*}
i.e., $(\mathbf{A}_6,\mathbf{B}_6)$ is a GCP.
\item
The proof of it is similar with 3) of Theorem \ref{thm-acp1}, so we omit it here.
\end{enumerate}
\end{proof}

\begin{Remark}
 $(\mathbf{A}_7,\mathbf{B}_{7})$  and $(\mathbf{A}_{14},\mathbf{B}_{14})$ in Theorem \ref{thm-gcs1} are two new constructions of ACPs of length $2^{m-1}+2^{m-2}$.
\end{Remark}

According to Lemma \ref{lem-bound}, Remark \ref{re-gdjpmepr} and Eq. (\ref{eq-pmeprN}), we can get an upper bound on the PMEPR of $\mathbf{A}_s$ in Theorem \ref{thm-gcs1}.
\begin{Corollary}
With the same notations as Theorem \ref{thm-gcs1}, we have
\begin{eqnarray*}
\hbox{PMEPR}(\mathbf{A}_s)&\leq&
\left\{\begin{array}{ll}
\frac{10}{3}, & s=7,14,\\
4, &s=3,12,\\
2, &s=1,2,4,6,8,15.\\
\end{array}\right.
\end{eqnarray*}
It is clear that
\begin{eqnarray*}
\hbox{PMEPR}_\mathcal{C}(\mathbf{a})\leq 4.
\end{eqnarray*}
\end{Corollary}

\begin{Example}\label{eg-2}
For $q=2$, let $\pi$ be a permutation of symbols $\{1,2,...,m\}$ with $\pi(k)=k$ for $1\leq k\leq m-2$ and $\pi(m-1)=m,\pi(m)=m-1$, and $a(\mathbf{x})=\sum_{k=1}^{m-1}x_{\pi(k)}
x_{\pi(k+1)}$. Table \ref{table-eg-2} shows the PMEPR properties of the contiguous subsequences of $\mathbf{a}$ for various $s$ and $m$. It can be observed that $\hbox{PMEPR}_\mathcal{C}(\mathbf{a})\leq4$, implying that  the PMEPRs of the contiguous subsequences in Example \ref{eg-1} are lower than those in Example \ref{eg-2}.
\begin{table}[htbp]
  \centering
    \caption{PMEPR comparison for the contiguous subsequences of $\mathbf{a}$ in Example \ref{eg-2} for $m=3,4,5,6$}\label{table-eg-2}
    {\small{
  \begin{tabular}{|c|c|c|c|c|c|c|c|c|c|c|}
  \hline
 \diagbox{$L=2^m$}{PMEPR}{$\mathbf{A}_s$}  &$\mathbf{A}_1$ & $\mathbf{A}_2$ & $\mathbf{A}_3$ & $\mathbf{A}_4$ & $\mathbf{A}_6$ & $\mathbf{A}_7$ & $\mathbf{A}_8$ & $\mathbf{A}_{12}$&$\mathbf{A}_{14}$&$\mathbf{A}_{15}$\\\hline
 8&2.0000 &2.0000 &4.0000 & 2.0000 & 1.7071 &2.6667&2.0000&2.0000& 1.6667& 2.0000\\
  \hline
   16&1.7071 & 1.7071 & 2.0000 & 1.7071 &2.0000 &3.0000&1.7071&3.4142 &3.3166& 1.7071
\\
  \hline
   32&2.0000 & 2.0000 &4.0000 &2.0000 & 1.8210&2.6667&2.0000&3.3066& 1.9369&2.0000\\
  \hline
   64&1.8210 & 1.8210 &3.4142 & 1.8210 & 2.0000 &3.1910&1.8210&3.6419&3.2424&1.8210\\
  \hline
\end{tabular}
}}
\end{table}

\end{Example}

Only the contiguous subsequences are considered in Theorems \ref{thm-acp1} and \ref{thm-gcs1}, while the non-contiguous frequency bands allocation also be widely used in OFDMA systems.
\section{The PMEPR Properties of Proposed Preamble Sequences for Non-contiguous Frequency bands Allocation}

In this section, we show the upper bounds of the PMEPR properties of the non-contiguous subsequences in Theorem \ref{thm-acp1} and  \ref{thm-gcs1} for the OFDMA system of non-contiguous DSA.

\begin{Theorem}\label{thm-non-acp1}
  With the same notations as Theorem \ref{thm-acp1}, let
\begin{eqnarray*}
 d(\mathbf{x}) &=& a(\mathbf{x})+ \frac{q}{2}x_{m}, \\
e(\mathbf{x}) &=& a(\mathbf{x})+ \frac{q}{2}x_{\pi(1)}+ \frac{q}{2}x_{m}.
\end{eqnarray*}
For $s=5,9,10,11,13$, $(\mathbf{A}_s,\mathbf{B}_s)$ has the following properties:
  \begin{enumerate}
    \item $(\mathbf{A}_{11},\mathbf{B}_{11})$ and $(\mathbf{A}_{13},\mathbf{B}_{13})$ are ACPs, i.e., for $s=11,13$,
\begin{eqnarray*}
|R_{\mathbf{A}_s}(\tau)+R_{\mathbf{B}_s}(\tau)|&=&
\left\{\begin{array}{ll}
6H, & \tau=0,\\
 2H, &\tau=2H,\\
0, & \text{otherwise}.
\end{array}\right.
\end{eqnarray*}
    \item $(\mathbf{A}_{9},\mathbf{B}_{9})$ is a GCP.
\item Both $(\mathbf{A}_5,\mathbf{B}_5,\mathbf{D}_5,\mathbf{E}_5)$ and $(\mathbf{A}_{10},\mathbf{B}_{10},\mathbf{D}_{10},\mathbf{E}_{10})$ are CSSs.
  \end{enumerate}
\end{Theorem}

\begin{proof}
For $s=11$, to prove $(\mathbf{A}_{11},\mathbf{B}_{11})$ is a GCP, we need to demonstrate that $R_{\mathbf{A}_{11}}(\tau)+R_{\mathbf{B}_{11}}(\tau)=0$ when $0<\tau\leq4H-1$.
\begin{itemize}
  \item When $0<\tau\leq H-1$, it can be obtained that
  \begin{eqnarray*}
    R_{\mathbf{A}_{11}}(\tau) &=& R_{\mathbf{A}_1}(\tau)+R^*_{\mathbf{A}_2,\mathbf{A}_1}(H-\tau)
    +R_{\mathbf{A}_2}(\tau)+R_{\mathbf{A}_8}(\tau),\\
      R_{\mathbf{B}_{11}}(\tau) &=& R_{\mathbf{B}_1}(\tau)+R^*_{\mathbf{B}_2,\mathbf{B}_1}(H-\tau)
    +R_{\mathbf{B}_2}(\tau)+R_{\mathbf{B}_8}(\tau).
  \end{eqnarray*}
   By Theorem \ref{thm-acp1}, $(\mathbf{A}_1,\mathbf{B}_1),(\mathbf{A}_2,\mathbf{B}_2),
  (\mathbf{A}_8,\mathbf{B}_8)$ are GCPs, so we have
  $$R_{\mathbf{A}_{11}}(\tau)+R_{\mathbf{B}_{11}}(\tau)=R^*_{\mathbf{A}_2,\mathbf{A}_1}(H-\tau)
  +R^*_{\mathbf{B}_2,\mathbf{B}_1}(H-\tau).$$
  For $0\leq i\leq H-1-(H-\tau)=\tau-1$, let $j=i+H-\tau$. Since $\pi(m-1)=m-1,\pi(m)=m$, then
  \begin{eqnarray*}
   \begin{array}{ll}
     A_1(\mathbf{x})=\frac{q}{2}\sum_{k=1}^{m-3}x_{\pi(k)}x_{\pi(k+1)}
     +\sum_{k=1}^{m-2}c_kx_k+c, & B_1(\mathbf{x})=A_1(\mathbf{x})+\frac{q}{2}x_{\pi(1)}, \\
     A_2(\mathbf{x})=A_1(\mathbf{x})+x_{\pi(m-2)}+c_{m-1}, & B_2(\mathbf{x})=A_2(\mathbf{x})+\frac{q}{2}x_{\pi(1)},  \\
  A_8(\mathbf{x})=A_1(\mathbf{x})+x_{\pi(m-2)}+c_{m-1}+c_m+\frac{q}{2}, & B_8(\mathbf{x})=A_8(\mathbf{x})+\frac{q}{2}x_{\pi(1)},  \end{array}
  \end{eqnarray*}
  so we have
  \begin{eqnarray*}
    A_2(i)-A_1(j) &=&\frac{q}{2}\sum_{k=1}^{m-3}\left(i_{\pi(k)}i_{\pi(k+1)}-j_{\pi(k)}j_{\pi(k+1)}\right)
     +\sum_{k=1}^{m-2}c_k\left(i_k-j_k\right)+\frac{q}{2}i_{\pi(m-2)}+c_{m-1}, \\
    B_2(i)-B_1(j) &=&A_2(i)-A_1(j)+\frac{q}{2}\left(i_{\pi(1)}-j_{\pi(1)}\right),
  \end{eqnarray*}
  and then
    \begin{eqnarray*}
  R^*_{\mathbf{A}_2,\mathbf{A}_1}(H-\tau)
  +R^*_{\mathbf{B}_2,\mathbf{B}_1}(H-\tau)&=&2\sum_{i\in J(H-\tau)}\xi_q^{-\left(A_2(i)-A_1(j)\right)},
\end{eqnarray*}
where $J(H-\tau)=\{0\le i\le H-1-(H-\tau): i_{\pi(1)}=j_{\pi(1)}\}$. For any $i\in J(H-\tau)$, let $t$ be the smallest integer such that $i_{\pi(t)}\ne j_{\pi(t)}$. which implies $2\leq t\leq m-2$.
Let $i'$ and $j'$ be integers which are different from $i$ and $j$ in only one position $\pi(t-1)$, i.e., $i'_{\pi(t-1)}=1-i_{\pi(t-1)}$, respectively, and so $j'=i'+\tau$. Then,
\begin{eqnarray*}
A_2(i)-A_1(j)-\left(A_2\left(i'\right)-A_1\left(j'\right)\right) \equiv\frac{q}{2}~~(\bmod~q),
\end{eqnarray*}
which implies
$\xi_q^{A_2(i)-A_1(j)}=(-1)\xi_q^{A_2(i')-A_1(j')}$. Hence,
\begin{eqnarray*}
 R^*_{\mathbf{A}_2,\mathbf{A}_1}(H-\tau)
  +R^*_{\mathbf{B}_2,\mathbf{B}_1}(H-\tau)=\sum_{i\in J(H-\tau)}\xi_q^{-\left(A_2(i)-A_1(j)\right)}
+\sum_{i'\in J(H-\tau)}\xi_q^{-\left(A_2(i')-A_1(j')\right)}=0,
\end{eqnarray*}
which implies $R_{\mathbf{A}_0}(\tau)+R_{\mathbf{B}_0}(\tau)=0$.
  \item When $H\leq\tau\leq2H-1$, we have
  \begin{eqnarray*}
    R_{\mathbf{A}_{11}}(\tau) &=& R_{\mathbf{A}_1,\mathbf{A}_2}(\tau-H),\\
      R_{\mathbf{B}_{11}}(\tau) &=& R_{\mathbf{B}_1,\mathbf{B}_2}(\tau-H).
  \end{eqnarray*}
  If $\tau=H$, since $\pi(m-1)=m-1,\pi(m)=m$, then $\pi(m-2)\leq m-2$ and
    \begin{eqnarray*}
    R_{\mathbf{A}_{11}}(H) = R_{\mathbf{B}_{11}}(H)= R_{\mathbf{A}_1,\mathbf{A}_2}(0)
    =\sum_{i=0}^{H-1}\xi_q^{-\frac{q}{2}i_{\pi(m-2)}-c_{m-1}}=\xi_q^{-c_{m-1}}
    \sum_{i=0}^{H-1}(-1)^{i_{\pi(m-2)}}=0.
  \end{eqnarray*}
  If $H<\tau\leq2H-1$, for $0\leq i\leq 4H-1-(\tau-H)$, let $j=i+\tau$, then we can get that
  \begin{eqnarray*}
    A_1(i)-A_2(j) &=&\frac{q}{2}\sum_{k=1}^{m-3}\left(i_{\pi(k)}i_{\pi(k+1)}-j_{\pi(k)}j_{\pi(k+1)}\right)
     +\sum_{k=1}^{m-2}c_k\left(i_k-j_k\right)-\frac{q}{2}j_{\pi(m-2)}-c_{m-1}, \\
    B_1(i)-B_2(j) &=&A_1(i)-A_2(j)+\frac{q}{2}\left(i_{\pi(1)}-j_{\pi(1)}\right).
  \end{eqnarray*}
  With similar arguments as the first case in this proof, it can also be obtained that
  \begin{eqnarray*}
    R_{\mathbf{A}_{11}}(\tau)+ R_{\mathbf{B}_{11}}(\tau) = R_{\mathbf{A}_1,\mathbf{A}_2}(\tau-H)+R_{\mathbf{B}_1,\mathbf{B}_2}(\tau-H)=0.
  \end{eqnarray*}
  \item When $2H\leq\tau\leq3H-1$, we have
  \begin{eqnarray*}
    R_{\mathbf{A}_{11}}(\tau) &=& R^*_{\mathbf{A}_8,\mathbf{A}_1}(H-(\tau-2H))+R_{\mathbf{A}_2,\mathbf{A}_8}(\tau-2H),\\
      R_{\mathbf{B}_{11}}(\tau) &=& R^*_{\mathbf{B}_8,\mathbf{A}_1}(H-(\tau-2H))+R_{\mathbf{B}_2,\mathbf{B}_8}(\tau-2H).
  \end{eqnarray*}
  If $\tau=2H$, then
   \begin{eqnarray*}
    R_{\mathbf{A}_{11}}(2H)= R_{\mathbf{B}_{11}}(2H)= R_{\mathbf{A}_6,\mathbf{A}_8}(0)=-H\xi_q^{-c_m},
  \end{eqnarray*}
  which implies $|R_{\mathbf{A}_{11}}(2H)+R_{\mathbf{B}_{11}}(2H)|=2H$.
  \item When $3H\leq\tau\leq4H-1$, it can be obtained that
   \begin{eqnarray*}
    R_{\mathbf{A}_{11}}(\tau) &=& R_{\mathbf{A}_1,\mathbf{A}_8}(\tau-3H),\\
      R_{\mathbf{B}_{11}}(\tau) &=& R_{\mathbf{B}_1,\mathbf{B}_8}(\tau-3H).
  \end{eqnarray*}
  If $\tau=3H$, we have
     \begin{eqnarray*}
    R_{\mathbf{A}_{11}}(3H)= R_{\mathbf{B}_{11}}(3H) = R_{\mathbf{A}_1,\mathbf{A}_8}(0)=\sum_{i=0}^{H-1}
    \xi_q^{-\frac{q}{2}i_{\pi(m-2)}-\frac{q}{2}-c_{m-1}-c_m}=0.
  \end{eqnarray*}
  If $3H<\tau\leq4H-1$, for $0\leq i\leq4H-1-\tau$, let $j=i+u$, then
   \begin{eqnarray*}
    A_1(i)-A_8(j) &=&\frac{q}{2}\sum_{k=1}^{m-3}\left(i_{\pi(k)}i_{\pi(k+1)}-j_{\pi(k)}j_{\pi(k+1)}\right)
     +\sum_{k=1}^{m-2}c_k\left(i_k-j_k\right)-\frac{q}{2}j_{\pi(m-2)}\\
     &&-\frac{q}{2}-c_{m-1}-c_m, \\
    B_1(i)-B_8(j) &=&A_1(i)-A_8(j)+\frac{q}{2}\left(i_{\pi(1)}-j_{\pi(1)}\right).
  \end{eqnarray*}
  With similar arguments as the first case in this proof, it can also be obtained that
  \begin{eqnarray*}
    R_{\mathbf{A}_{11}}(\tau)+ R_{\mathbf{B}_{11}}(\tau) = R_{\mathbf{A}_1,\mathbf{A}_8}(\tau-3H)+R_{\mathbf{B}_1,\mathbf{B}_8}(\tau-3H)=0.
  \end{eqnarray*}
\end{itemize}
We can similarly prove the results of $s=5,9,10,13$, so we omit it here.
\end{proof}

\begin{Remark}
The constructions of GCPs, CSSs and ACPs in Theorem \ref{thm-non-acp1} are new for length $2^m$ and spectral nulls inside.
\end{Remark}

The following lemma is straightforward from Lemma \ref{lem-bound} and Eq. (\ref{eq-pmeprN}).

\begin{Corollary}\label{cor-ncpmepr1}
With the same notations as Theorem \ref{thm-non-acp1}, for $s=5,9,10,11,13$, we have
\begin{eqnarray*}
\hbox{PMEPR}(\mathbf{A}_s)&\leq&
\left\{\begin{array}{ll}
\frac{10}{3}, & s=11,13,\\
2, &s=9,\\
4, &s=5,10,
\end{array}\right.
\end{eqnarray*}
which implies that $\hbox{PMEPR}_{\mathcal{NC}}(\mathbf{a})\leq4$.
\end{Corollary}

\begin{Example}\label{eg-nc1}
With notations in Example \ref{eg-1}, Table \ref{table-eg-nc1} shows the PMEPR properties of the non-contiguous subsequences of sequence $\mathbf{a}$ for various $m$ and $s$. It can be observed that
\begin{itemize}
  \item for $L=8$, $\hbox{PMEPR}_{\mathcal{NC}}(\mathbf{a})=4.0000$;
  \item for $L=16$, $\hbox{PMEPR}_{\mathcal{NC}}(\mathbf{a})=3.4142$;
  \item for $L=32$, $\hbox{PMEPR}_{\mathcal{NC}}(\mathbf{a})=4.0000$;
  \item for $L=64$, $\hbox{PMEPR}_{\mathcal{NC}}(\mathbf{a})=3.6419$.
\end{itemize}
 Hence, Theorem 3 is verified.
\begin{table}[ht]
  \centering
\caption{PMEPR comparison of the non-contiguous subsequences for various $m$ and $s$ in Example \ref{eg-nc1}}\label{table-eg-nc1}
\begin{tabular}{|c|c|c|c|c|c|}
  \hline
  \diagbox{$L=2^m$}{$\hbox{PMEPR}(\mathbf{A}_s)$}{$\mathbf{A}_s$} & $\mathbf{A}_5$ &$\mathbf{A}_9$  &   $\mathbf{A}_{10}$ &$\mathbf{A}_{11}$& $\mathbf{A}_{13}$ \\ \hline
  $8$&  4.0000 & 1.7071  & 3.4142& 1.9024  &  2.6667   \\ \hline
  $16$& 3.4142 & 2.0000 &   3.3066&  3.1910 & 3.3166  \\ \hline
  $32$&   4.0000&1.8210 &    3.4765&3.2077  & 3.3166\\ \hline
  $64$&  3.6419& 2.0000  &    3.6029 & 3.3158 &  3.3166 \\ \hline
\end{tabular}
\end{table}
\end{Example}

\begin{Theorem}\label{thm-non-gcs1}
  With the same notations as Theorem \ref{thm-gcs1}, for $s=5,9,10,11,13$, $(\mathbf{A}_s,\mathbf{B}_s)$ has the following properties:
  \begin{enumerate}
    \item $(\mathbf{A}_{11},\mathbf{B}_{11})$ and $(\mathbf{A}_{13},\mathbf{B}_{13})$ are ACPs, i.e., for $s=11,13$,
\begin{eqnarray*}
|R_{\mathbf{A}_s}(\tau)+R_{\mathbf{B}_s}(\tau)|&=&
\left\{\begin{array}{ll}
6H, & \tau=0,\\
 2H, &\tau=H;\\
0, & \text{otherwise}.
\end{array}\right.
\end{eqnarray*}
    \item $(\mathbf{A}_5,\mathbf{B}_5)$, $(\mathbf{A}_9,\mathbf{B}_9)$ and $(\mathbf{A}_{10},\mathbf{B}_{10})$ are GCPs.
  \end{enumerate}
\end{Theorem}
\begin{proof}
 The proof here is similar with it of Theorem \ref{thm-non-acp1}, so we omit it here.
\end{proof}

\begin{Remark}
The constructions of GCPs, and ACPs in Theorem \ref{thm-non-gcs1} are new for lengths $2^m$ and spectral nulls inside.
\end{Remark}

As Corollary \ref{cor-ncpmepr1}, we get an upper bound on the PMEPR of $\mathbf{A}_s$ in Theorem \ref{thm-non-gcs1}.

\begin{Corollary}\label{cor-ncpmepr2}
With the same notations as Theorem \ref{thm-non-gcs1}, for $s=5,9,10,11,13$, we have
\begin{eqnarray*}
\hbox{PMEPR}(\mathbf{A}_s)&\leq&
\left\{\begin{array}{ll}
\frac{10}{3}, & s=11,13;\\
2, &s=5,9,10\\
\end{array}\right.
\end{eqnarray*}
which implies that $\hbox{PMEPR}_{\mathcal{NC}}(\mathbf{a})\leq\frac{10}{3}$.
\end{Corollary}

\begin{Example}\label{eg-nc2}
With notations in Example \ref{eg-2}, Table \ref{table-eg-nc2} shows the PMEPR properties of the non-contiguous subsequences of sequence $\mathbf{a}$ for various $m$ and $s$. It can be observed that
\begin{itemize}
  \item for $L=8$, $\hbox{PMEPR}_{\mathcal{NC}}(\mathbf{a})=2.6667$;
  \item for $L=16$, $\hbox{PMEPR}_{\mathcal{NC}}(\mathbf{a})=3.0000$;
  \item for $L=32$, $\hbox{PMEPR}_{\mathcal{NC}}(\mathbf{a})=3.3166$;
  \item for $L=64$, $\hbox{PMEPR}_{\mathcal{NC}}(\mathbf{a})=3.3166$.
\end{itemize}
With these explanations, we verify Theorem 4.
\begin{table}[ht]
  \centering
\caption{PMEPR comparison of the non-contiguous subsequences for various $m$ and $s$ in Example \ref{eg-nc2}}\label{table-eg-nc2}
\begin{tabular}{|c|c|c|c|c|c|}
  \hline
  \diagbox{$L=2^m$}{$\hbox{PMEPR}(\mathbf{A}_s)$}{$\mathbf{A}_s$} & $\mathbf{A}_5$ & $\mathbf{A}_9$ &  $\mathbf{A}_{10}$&$\mathbf{A}_{11}$& $\mathbf{A}_{13}$  \\ \hline
  $8$&    1.0000 & 1.7071 &     1.0000 & 2.6667  &   1.6667   \\ \hline
  $16$&   2.0000 &2.0000 &     2.0000& 3.0000 &   1.9084  \\ \hline
  $32$&  1.7682 & 1.8210 &     1.7682 &3.3166  & 3.1910\\ \hline
  $64$&  2.0000  &  2.0000 &    2.0000&3.3166 &  3.1157 \\ \hline
\end{tabular}
\end{table}
\end{Example}
\begin{Remark}
Combining Corollaries 1--5, it can be obtained that the PMEPR of the contiguous and non-contiguous subsequences of GDJ sequences introduced in this paper are less than 4.
Note that the reduction of PMEPRs of proposed sequences is not because the increase of the average power. As defined in equation (3), for any GDJ sequence $\mathbf{a}$ and a given subcarrier index set $\Omega$, the average power is equal to the cardinality of $\Omega$,  i.e.,
\begin{eqnarray*}
 P_{av}({\mathbf{\tilde{a}}})=R_{\mathbf{\tilde{a}}}(0)=|\Omega|,
 \end{eqnarray*}
 where $\mathbf{\tilde{a}}$ is the complex-valued version of $\mathbf{a}$ corresponding to $\Omega$.
\end{Remark}
\begin{Remark}
The proposed sequences are constructed in a  systematic way: First, generate a sequence $\psi(\mathbf{a})$ from Theorem 1 or 2 depending on the number of subcarriers. Then, for a given subcarrier index set $\Omega$ allocated to this system, null the elements of $\psi(\mathbf{a})$ whose indices are not in $\Omega$ to zero.
\end{Remark}
\section{PMEPR Comparisons}

In this section, we compare the proposed
preamble sequences are undertaken, showing their good PMEPR properties for DSA OFDMA systems compared with $m$-sequences and Zadoff-Chu (ZC) sequences in terms of their PMEPR performances under DSA transmissions.  In our comparisons, we consider enlarged ZC sequence and $m$-sequence of length $64$ and $32$ by adding ``$-1$'' at the end of a ZC sequence and $m$-sequence, respectively. The parameters of ZC sequence and $m$-sequence are chosen from \cite{3GPP-17}, where the ZC sequence root index is 25. The proposed sequences in our comparisons are chosen from Examples \ref{eg-1} and \ref{eg-2}.

Tables \ref{table-comparison-32} and \ref{table-comparison-64} list the PMEPRs of some subsequences of $m$-sequences, ZC sequences and proposed preamble sequences, where ``Proposed Sequence Family $\mathcal{X}$'' refers to the proposed sequences in Theorem \ref{thm-acp1}, and ``Proposed Sequence Family $\mathcal{Y}$'' refers to the proposed sequences in Theorem \ref{thm-gcs1}. It can be observed that the PMEPRs of $\mathbf{A}_2$, $\mathbf{A}_9$ and $\mathbf{A}_{15}$ are significantly lower than those of  $m$-sequences and ZC sequences both in Tables \ref{table-comparison-32} and \ref{table-comparison-64}, which are 2 at most, while for $m$-sequences and ZC sequences, the PMEPRs of their subsequences can be as high as 4.5000 and 3.7842. In Table \ref{table-comparison-32}, the PMEPRs of $\mathbf{A}_{14}$ are lower than those of $m$-sequences and ZC sequences, while in Table \ref{table-comparison-64} they are slightly higher than those of $m$-sequences and ZC sequences. However, our proposed preamble sequences outperform  $m$-sequences and ZC sequences with regard to the maximum PMEPR of all the contiguous and non-contiguous subsequences.
\begin{table}[ht]
  \centering
\caption{PMEPR comparisons among $m$-sequences, ZC sequences and proposed sequences of length $L=32$  in DSA OFDMA systems}\label{table-comparison-32}
{\small
\begin{tabular}{|c|c|c|c|c|c|c|c|}
  \hline
  \diagbox[width=10em,trim=l]{Sequences}{$\hbox{PMEPR}(\mathbf{A}_s)$}{$\mathbf{A}_s$} & $\mathbf{A}_2$  &  $\mathbf{A}_{9}$&$\mathbf{A}_{14}$& $\mathbf{A}_{15}$&$\hbox{PMEPR}_{\mathcal{C}}(\mathbf{a})$& $\hbox{PMEPR}_{\mathcal{NC}}(\mathbf{a})$ & $\hbox{PMEPR}_{\mathcal{A}}(\mathbf{a})$ \\ \hline
  ZC sequence\cite{3GPP-17}&   2.8072 & 3.7842 & 3.6073  &  2.4250 &3.6313 & 4.0079 &4.0079\\ \hline
  $m$-sequence\cite{3GPP-17}&   4.5000  &   3.1269& 3.3333  & 2.2500  & 4.5000&3.1269&4.5000\\ \hline
  Proposed Sequence& \multirow{2}*{2.0000}  & \multirow{2}*{1.8210}&\multirow{2}*{3.3274}& \multirow{2}*{2.0000}&\multirow{2}*{3.3274}&\multirow{2}*{4.0000}&\multirow{2}*{4.0000}\\ Family $\mathcal{X}$&&&&&&&\\ \hline
  Proposed Sequence&  \multirow{2}*{2.0000}&\multirow{2}*{1.8210}&\multirow{2}*{1.9369} &\multirow{2}*{2.0000} &\multirow{2}*{4.0000}&\multirow{2}*{3.3166}&\multirow{2}*{4.0000}\\
   Family $\mathcal{Y}$&&&&&&&\\ \hline
\end{tabular}}
\end{table}
\begin{table}[ht]
  \centering
\caption{PMEPR comparisons among $m$-sequences, ZC sequences and proposed sequences of length $L=64$  in DSA OFDMA systems}\label{table-comparison-64}
{\small
\begin{tabular}{|c|c|c|c|c|c|c|c|}
  \hline
  \diagbox[width=10em,trim=l]{Sequences}{$\hbox{PMEPR}(\mathbf{A}_s)$}{$\mathbf{A}_s$} & $\mathbf{A}_2$  &  $\mathbf{A}_{9}$&$\mathbf{A}_{14}$& $\mathbf{A}_{15}$&$\hbox{PMEPR}_{\mathcal{C}}(\mathbf{a})$& $\hbox{PMEPR}_{\mathcal{NC}}(\mathbf{a})$ &$\hbox{PMEPR}_{\mathcal{A}}(\mathbf{a})$  \\ \hline
  ZC sequence\cite{3GPP-17}&   3.4609 & 3.4519 & 2.7979  &  2.7952 &4.6421 & 4.1302 &4.6421 \\ \hline
  $m$-sequence\cite{3GPP-17}&   2.8762  &  3.5175& 2.6213  & 2.2101  & 5.1374&3.5175&5.1374\\ \hline
 Proposed Sequence&\multirow{2}*{ 1.8210}  &  \multirow{2}*{ 2.0000}& \multirow{2}*{2.9419}  & \multirow{2}*{1.8210}&\multirow{2}*{3.1910}&\multirow{2}*{3.6419}&\multirow{2}*{3.6419}\\
   Family $\mathcal{X}$&&&&&&&\\ \hline
 Proposed Sequence& \multirow{2}*{1.8210}   & \multirow{2}*{2.0000}&\multirow{2}*{3.2424} & \multirow{2}*{1.8210} &\multirow{2}*{3.6419}&\multirow{2}*{3.3166}&\multirow{2}*{3.6419}\\
   Family $\mathcal{Y}$&&&&&&&\\ \hline
\end{tabular}}
\end{table}
\section{Conclusion}

In this paper, we have introduced two classes of preamble sequences from GDJ sequences for PMEPR reduction of OFDMA systems with DSA, specifically, over contiguous or non-contiguous spectral sub-bands which are carved from four adjacent RBs.  In the first class, the subsequences corresponding to contiguous and non-contiguous DSAs have PMEPRs upper bounded by 3.3334 and 4, respectively. On the other hand, the second class consists of subsequences corresponding to contiguous and non-contiguous DSAs have PMEPRs upper bounded by 4 and 3.3334, respectively. Compared with $m$-sequences and Zadoff-Chu sequences, our proposed sequences have better PMEPR properties for any OFDMA DSA schemes. We remark that the same PEMPR properties may be hard to attain when the number of RBs is larger than four. For example, it can be observed by concatenating the proposed preamble sequences $k$ times to an OFDMA system with $4k$-RBs, to form a $4k$-RB DSA OFDMA system, the resultant subsequence PMEPRs may be as high as $4k$. How to construct preamble sequences having low PMEPRs for any number of RBs  (different from four) will be a challenging but interesting future direction of this research.

\bibliographystyle{IEEEtran}
\bibliography{IEEEfull}

\end{document}